\documentclass{article}

\usepackage{PRIMEarxiv}

\usepackage[utf8]{inputenc} 
\usepackage[T1]{fontenc}    
\usepackage{hyperref}       
\usepackage{url}            
\usepackage{booktabs}       
\usepackage{amsfonts}       
\usepackage{nicefrac}       
\usepackage{microtype}      
\usepackage{lipsum}
\usepackage{fancyhdr}       
\usepackage{graphicx}       
\graphicspath{{media/}}     

\usepackage{longtable}
\usepackage{soul}
\usepackage{xcolor}
\usepackage{float}
\usepackage{amsmath,amssymb,amsthm}
\usepackage{iftex}

\usepackage{graphicx}
\usepackage{float}

\usepackage{natbib}

\newtheorem{theorem}{Theorem}[subsection]
\newtheorem{corollary}{Corollary}[theorem]

\newtheorem{proposition}[theorem]{Proposition}

\newtheorem{assumption}{Assumption}
  
\title{Near-Optimal Nitrogen Recommendations for Precision Agriculture via Sequential Screening and Hierarchical Refinement}

\author{
  Sakshi Arya \\
  Department of Mathematics, Applied Mathematics and Statistics,\\
Case Western Reserve University, \\
Cleveland, Ohio\\\\
  \texttt{sxa1351@case.edu} \\
   \And
 Abdul-Nasah Soale \\
 Department of Mathematics, Applied Mathematics and Statistics,\\
Case Western Reserve University, \\
Cleveland, Ohio\\\\
  \texttt{axs2953@case.edu} \\
   \And
  Hossein Moradi Rekabdarkolaee\\
Business Analytics, Economics, and Information System \\
  Bowling Green State University,\\
   Bowling Green, Ohio\\
    \texttt{hmrekab@bgsu.edu} \\
}

\begin{document}
\maketitle

\begin{abstract}
Nitrogen fertilizer management plays a central role in balancing agricultural productivity and environmental sustainability, yet identifying optimal application strategies remains difficult because treatment responses vary substantially across locations and many fertilizer choices are statistically indistinguishable near the optimum. This paper develops a hierarchical refinement procedure, built on sequential screening, for fertilizer recommendation in multi-site experiments that explicitly accounts for spatial heterogeneity while prioritizing parsimonious, decision-oriented selection. Rather than targeting a single estimated best treatment, the proposed method first conducts sequential screening at a higher aggregation level to eliminate clearly inferior fertilizer choices and then refines recommendations locally among the surviving candidates. We study the asymptotic
properties of the proposed estimators and show that it provides screening-safety guaranteed recommendations. The  efficacy of the new approach is investigated through a multi-state, multi-year corn nitrogen trial. The results show that no single fertilizer regime is uniformly optimal within a state; instead, each state is associated with multiple recommended choices, and the most common recommendation typically covers only about one-third to one-half of decision units, underscoring substantial within-state heterogeneity. Representative site-level comparisons further demonstrate that the proposed method often yields lower total nitrogen recommendations than state-level or hindsight benchmarks while maintaining competitive agronomic performance. 
\end{abstract}


\section{Introduction}
Nitrogen (N) fertilizers are central to modern agriculture because they supply a nutrient that often limits crop growth, enabling large yield gains and underpinning global food security \citep{smil2004enriching, galloway2008transformation, cassman2002agroecosystems}. The widespread adoption of synthetic N produced via the Haber-Bosch process has allowed farmers to decouple crop production from local organic N sources, supporting the intensification of cereal systems and contributing substantially to the observed increases in calorie availability over the last century \citep{smil2004enriching, dobermann2005nitrogen, galloway2008transformation}. However, only a fraction of applied N is recovered in harvested products. The global N use efficiency commonly averages around 30 to 40\% and the remainder lost to the environment through leaching, gaseous emissions, and accumulation in soil and water bodies \citep{cassman2002agroecosystems, raun2002improving, lassaletta201450}. These inefficiencies have driven growing concern that current N use patterns exceed regional and planetary environmental limits \citep{sutton2011european, steffen2015planetary, cameron2013nitrogen, shcherbak2014global}. Furthermore, excess fertilizer N drives non-linear increases in nitrous oxide $(N_2O)$ emissions, a greenhouse gas with a high global warming potential and significant ozone-depleting effect \citep{shcherbak2014global, reay2012global}. Ammonia volatilization from urea- and ammonium-based fertilizers further contributes to particulate air pollution and can damage sensitive vegetation \citep{sutton2013towards, guthrie2018impact}. These reinforce the need to identify agronomically optimal but environmentally sustainable N application rates and management strategies 

Traditionally, the ``optima'' amount of N is determine by balancing yield response against economic and environmental costs using several complementary approaches. Classical agronomy uses N response trials and quadratic or quadratic‑plateau models to estimate the economic optimum N rate (EONR), where the marginal value of yield gain equals the marginal cost of added N \citep{cassman2002agroecosystems, sawyer2006concepts}. At field and farm scales, partial N budgets and whole‑farm N balances are used to assess whether inputs exceed off take, providing a diagnostic of long‑term N surpluses and leaching risk \citep{cherry2012using, tamagno2022quantifying}. Increasingly, ``4R'' (Right Source, Right Rate, Right Time, and Right Place) strategies and N use efficiency indicators, such as recovery efficiency and partial factor productivity, guide rate decisions, often supported by soil tests, canopy sensors, or leaf color charts to better synchronize N supply with crop demand in real time \citep{raun2002improving, dobermann2005nitrogen}. Meta‑analyses and global syntheses show that combining these tools with site‑specific management can maintain high yields while substantially reducing N losses and narrowing the gap between current practice and planetary or regional N limits \citep{steffen2015planetary, cameron2013nitrogen, lassaletta201450}.

A growing body of literature has applied statistical and machine learning (ML) methods to improve the prediction of optimal N rates across variable field conditions. Classical approaches include linear, quadratic, quadratic-plateau, and linear-plateau response functions has been fitted to site-year N response trial data to EONR \citep{cerrato1990comparison, bullock1994quadratic}. Mixed-effects models and Bayesian hierarchical frameworks have been used to pool information across sites and years while accounting for spatial and temporal variability in yield response to N \citep{archontoulis2014methodology, puntel2018systems}. More recently, ensemble machine learning methods such as random forests, gradient boosting machines, and support vector regression have been applied to predict site-specific EONR from soil, weather, and management covariates, often outperforming classical regression in cross-validation \citep{ransom2019statistical, qin2018application, de2023predicting}. Deep learning approaches have shown to be promising for capturing the temporal dynamics of N uptake and loss \citep{zhang2024data}. A key challenge is the limited transferability of trained models across environments. There are many ongoing work on transfer learning and domain adaptation to improve generalizability. 

More recently, sequential decision-making methods have begun to attract attention in precision agriculture. Reinforcement learning and multi-armed bandit algorithms provide principled frameworks for balancing learning and decision-making under uncertainty and have been explored for adaptive management and resource optimization in agricultural systems \citep{saikai2018multi, saikai2020machine, gautron2022reinforcement, huang2025precision, Arya2026}. Related work has also investigated Gaussian process regression and Bayesian optimization as principled frameworks for adaptive N management, enabling sequential updating of rate recommendations as new in-season data become available \citep{mamo2003spatial, ng2022bayesian, matavel2025bayesian}. While these approaches typically focus on deciding management options for maximizing yield or EONR, they generally seek a single optimal recommendation for each environment.

Most statistical and ML approaches to N management emphasize predicting crop yield or economically optimal N rates, typically with the goal of selecting a single best recommendation for each field or environment. However, the highest-yielding treatment may not be the most economically or environmentally desirable. In this paper, instead of strict best-treatment selection, we focus on identifying practically near-optimal fertilizer options and favoring lower-N choices when yield differences are negligible. We propose a hierarchical refinement framework built on sequential screening, in which sequential screening is used to eliminate clearly inferior fertilizer treatments, and recommendations are subsequently refined across spatial scales by selecting among the surviving near-optimal options. This method is conceptually related to ranking-and-selection procedures for identifying the best or near-best alternatives \citep{bechhofer1954single, kim2001fully} and best-arm identification and successive elimination methods studied in stochastic bandits \citep{even2006action, kalyanakrishnan2012pac}. However, unlike these adaptive-sampling approaches, our objective is not to develop a new ranking-and-selection algorithm, but rather to demonstrate how screening-based subset identification can be incorporated into a hierarchical decision framework for spatially heterogeneous agricultural systems using historical experimental data. Treatments that appear clearly suboptimal are eliminated sequentially, and final recommendations are obtained by selecting the lowest-N fertilizer choice from the surviving near-optimal set. 

The rest of the paper is organized as follows: Section \ref{sec: data} describes the data. Section \ref{method} presents the proposed approach and Section \ref{real} contains the results from analyzing real data. We conclude the article with a short discussion and suggestions for further work in Section \ref{sec: discussion}. Additional results are provided in the Appendix.

\section{Data} 
\label{sec: data}

The data \citep{ransom2021data} analyzed in Section \ref{real} is obtained from a coordinated public–industry corn N research network conducted across the eight states (Iowa, Illinois, Indiana, Minnesota, Missouri, North Dakota, Nebraska, Wisconsin) following a common protocol to ensure comparability across locations. Protocol development specified project organization, investigator roles, site selection criteria, experimental design, N fertilization treatments and implementation, use of shared equipment, sampling schedules, and data management procedures. Figure \ref{fig1} shows the location of the fields on US map. The dataset consists of multi-year agricultural field experiments conducted from 2014--2016 under a randomized complete block design (RCBD). Over the course of the study, experiments were conducted across 31 unique agricultural sites. A trial corresponds to a unique site-year combination and represents a single experimental field observed during a particular growing season. In total, the dataset contains 49 site-year trials. Within each trial, the field was partitioned into four blocks representing relatively homogeneous experimental units, resulting in 195 trial-block decision units across the study. Fertilizer treatments were assigned within each block, and the resulting grain yield was measured at the end of the growing season. The goal is to determine whether fertilizer recommendations should be made globally, at the state level, or through a hierarchical state-to-site refinement strategy, while balancing grain yield and N usage. More information about the data and its collection process can be found in \cite{kitchen2017public} and references therein. 

\begin{figure}
    \centering
    \includegraphics[width=0.9\linewidth]{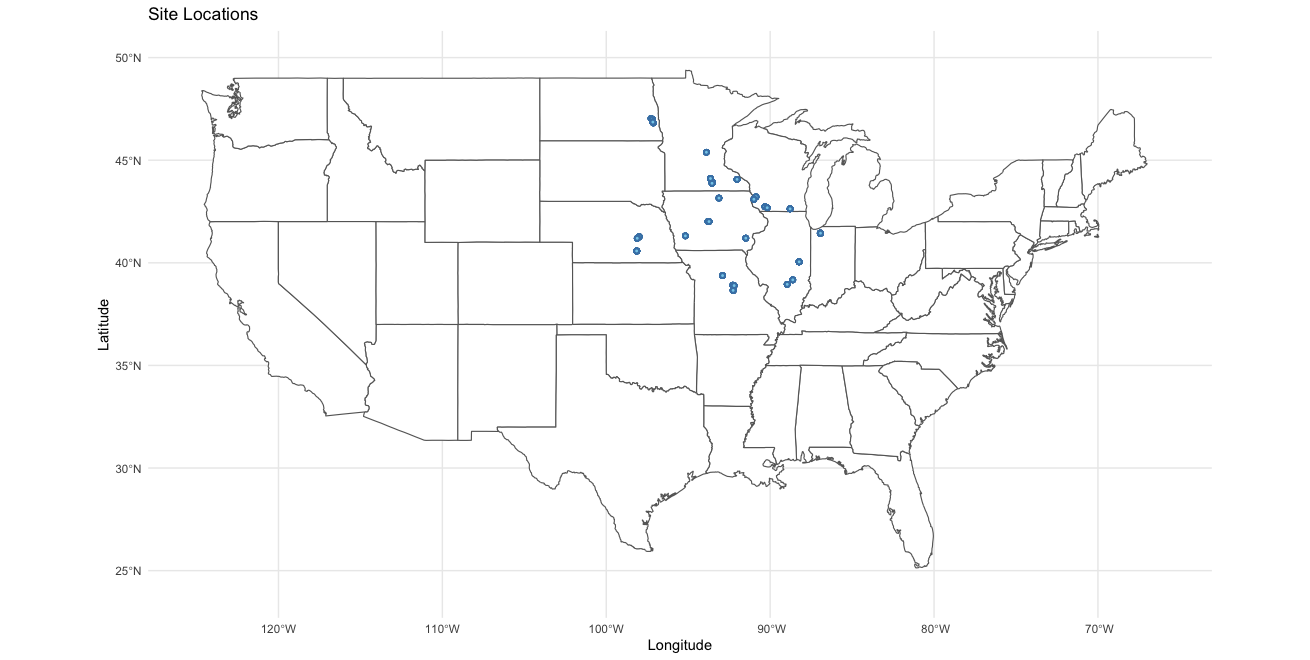}
    \caption{Location of the data at United States' Midwest Corn Belt.}
    \label{fig1}
\end{figure}

Primary exploratory data analyses using the full dataset was performed to characterize the structure of the N response surface and assess the extent of contextual treatment heterogeneity. Figure~\ref{fig:yield-total-n} summarizes the empirical relationship between total N application and grain yield. Mean yield increases sharply at low N levels and then exhibits a relatively flat response at moderate-to-high N levels. This pattern suggests a concave N response surface with diminishing marginal gains at larger N application rates. Consequently, several fertilizer arms may produce statistically similar yields despite substantially different N usage levels. This observation motivates both the inclusion of nonlinear N-response terms in the reward model and the use of near-best-arm screening procedures rather than strict best-arm selection. This study is motivated by two practical questions in N management:
\begin{enumerate}
    \item Can a single fertilizer recommendation be used across the Midwest, or should recommendations vary across regions and production environments? 
    \item When a recommended fertilizer amount achieves near-maximal yield, are there alternative N treatments that attain similar yield outcomes while using less N?
\end{enumerate}
To address these questions, we compare a collection of recommendation and screening procedures that operate at different spatial scales, ranging from globally pooled recommendations to state-level and site-refined strategies. The goal is to identify fertilizer management programs that balance grain yield performance with reduced N usage while accounting for spatial heterogeneity across production environments.

\begin{figure}
    \centering
    \includegraphics[width=0.7\textwidth]{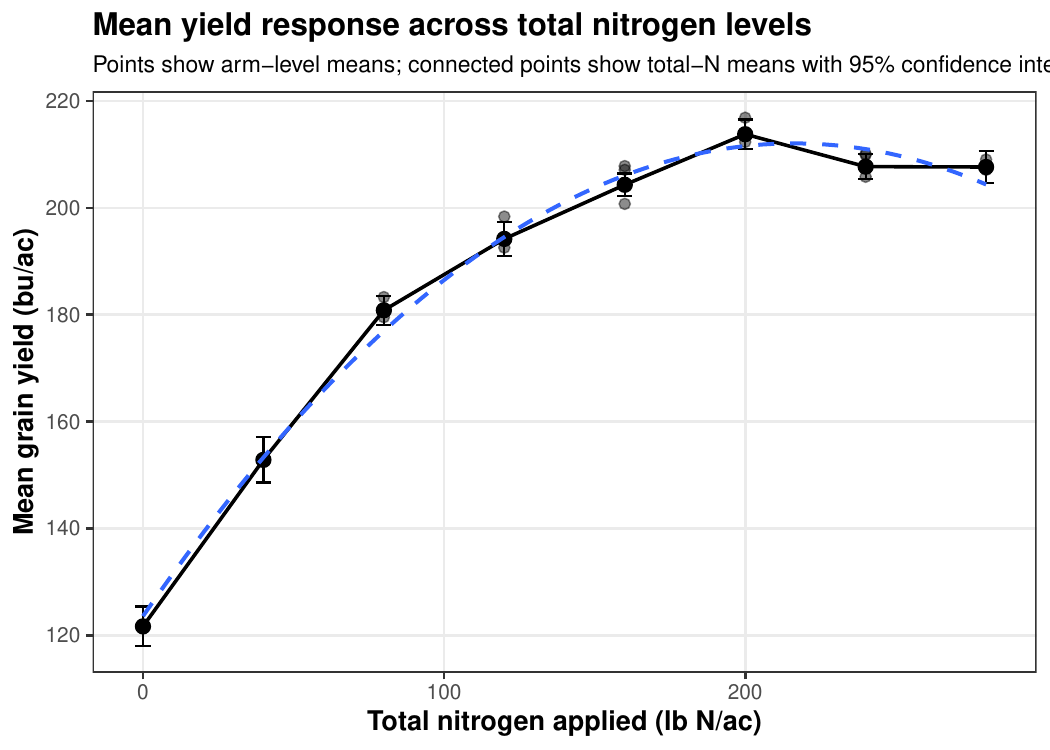}
    \caption{Mean grain yield as a function of total N application. Points show arm-level means, connected points show total-N means, blue dashed line shows a quadratic fit, and error bars denote 95\% confidence intervals.}
    \label{fig:yield-total-n}
\end{figure}

\section{Method}
\label{method}

The procedures considered in this study fall into three categories. First, we consider retrospective recommendation benchmarks that use the full dataset and therefore represent idealized recommendations rather than deployable sequential strategies. Second, we consider yearly adaptive recommendations that update fertilizer choices using only information from previous years. Finally, we consider sequential screening procedures  for arm elimination, followed by a hierarchical refinement step, whose goal is to tailor recommendations across spatial scales by selecting among the surviving near-optimal fertilizer choices. 

\subsection{Notation}
For the purpose of recommendation and policy evaluation, each trial-block combination is treated as a single decision unit. Let $s$ denote the site-year trial index and let $b$ denote the block index within the trial. Each decision unit therefore corresponds to $t := (s,b)$. Each fertilizer treatment (or \emph{arm}) is defined by a combination of planting N ($N^{\text{plant}}$) and sidedress N ($N^{\text{side}}$) application rates, $a = (N^{\text{plant}}, N^{\text{side}})$, with total N application given by $N(a) = N^{\text{plant}} + N^{\text{side}}$. Across the study, a total of 16 fertilizer arms are considered. The primary response variable is grain yield measured in bushels per acre. Available pre-treatment covariates include soil nitrate measurements, site productivity classification, and planting population.

Let $m=1,\ldots,M$ index the experimental years, where each year forms a \emph{batch}. Within batch $m$, decision units are observed sequentially and indexed by $t=1,\ldots,T_m$, where each decision unit corresponds to a trial-block combination. At decision unit $t$ in batch $m$, a recommendation procedure observes a context vector $X_{m,t}\in \mathcal{X} \subseteq \mathbb R^p,$
containing pre-treatment agronomic covariates such as soil nitrate measurements and site productivity indicators. The  recommendation procedure must then choose a fertilizer arm $a_{m,t}\in\mathcal A,$ where $\mathcal A$ denotes the set of 16 N management programs. After selecting arm $a_{m,t}$, the learner eventually observes a \emph{reward}, $Y_{m,t}(a_{m,t}),$ corresponding to the realized grain yield for that decision unit. However, these outcomes are not observed immediately. Instead, all yield outcomes from batch $m$ become available only after completion of the growing season. Thus, the information available before batch $m+1$ is
\begin{align*}
\mathcal H_m = \left\{ (X_{k,t},a_{k,t},Y_{k,t}(a_{k,t})) : 1\le k\le m \right\}.
\end{align*}
The primary objective is to identify N recommendations that achieve high yield while reducing N usage. Since treatment response varies across states and sites and multiple fertilizer programs often produce similar yields, a single globally optimal recommendation may not be appropriate. Instead, we focus on identifying near-optimal fertilizer treatments and selecting lower-N alternatives whenever yield differences are practically negligible. To assess the value of regional adaptation and spatial refinement, we compare recommendation strategies ranging from globally pooled rules to state-level and hierarchical state-to-site screening procedures.

\subsection{Retrospective Recommendation Benchmarks}
These procedures use the full dataset and serve as retrospective benchmarks for understanding the value of global, state-specific, and site-specific recommendations.

\subsubsection{Global fixed-arm recommendation}
The simplest benchmark ignores all spatial and temporal information and recommends a single fertilizer arm for every decision unit. Let
\begin{align*}
\widehat{\mu}(a) = \frac{1}{n_a} \sum_{t: a \text{ observed at } t} Y_t(a)
\end{align*}
denote the empirical mean yield for arm $a$ across all states, sites, years, and blocks. The global fixed-arm policy selects the arm as
\begin{align*}
\widehat a_{\mathrm{global}} = \arg\max_{a\in\mathcal A} \widehat{\mu}(a),
\end{align*}
and recommends it for every decision unit, regardless of year, state, site, or covariates.

\subsubsection{State-specific fixed-arm recommendation.}
The second benchmark allows recommendations to vary by state, but not by site, year, or block. For each state $q$, define
\begin{align*}
\widehat{\mu}_q(a) = \frac{1}{n_{q,a}} \sum_{t: \mathrm{State}(t)=q,\; a \text{ observed at } t} Y_t(a).
\end{align*}
The state-specific fixed-arm policy selects
\begin{align*}
\widehat a_q = \arg\max_{a\in\mathcal A} \widehat{\mu}_q(a),
\end{align*}
and recommends $\widehat a_q$ for every decision unit in state $q$.

\subsubsection{Retrospective Hindsight Benchmarks}
We also included the same-year hindsight benchmarks to quantify optimistic upper limits. These policies choose the best fertilizer choice using outcomes from the same year being evaluated. For example, define
\begin{align*}
\widehat{\mu}_{q,m}^{\mathrm{same}}(a) = \frac{1}{n_{q,m,a}} \sum_{t: \ \mathrm{State}(t)=q, \ \mathrm{Year}(t)=m,\; a\text{ observed at }t} Y_t(a),
\end{align*}
where $n_{q,m,a}$ denotes the number of observations of fertilizer choice $a$ in state $q$ during year $m$. The state-year hindsight benchmark then selects
\begin{align*}
\widehat a^{\mathrm{hind}}_{q,m} = \arg\max_{a\in\mathcal A} \widehat{\mu}_{q,m}^{\mathrm{same}}(a).
\end{align*}
We analogously define global-year and site-year hindsight benchmarks. These are not deployable sequential policies, since they use outcomes from the same year being evaluated, but they provide useful references for the maximum performance that could be achieved with perfect same-year information.

The recommendations produced by these retrospective benchmarks are subsequently applied to all trial-block decision units in the dataset and evaluated using the policy evaluation criteria. It is worth noting that this scenario is not applicable to real-world conditions, because the outcome remains unknown until harvest, which takes place after N application. The sole purpose of this benchmark is to establish an upper bound on the performance that such a scenario could theoretically achieve.

\subsubsection{Cumulative Yearly Adaptive Recommendations}
These policies update recommendations between years using only data from previous years. This approach reflects a practical decision-making process, whereby farmers adjust N application rates based on the harvested yield from the previous season. One can have this policies as a global or state level. For the global scenario, let
\begin{align*}
\mathcal H_{m-1} = \{t:\text{ decision unit }t\text{ occurred before year }m\}
\end{align*}
denote the cumulative history available before making recommendations in year $m$. The global yearly adaptive policy selects
\begin{align*}
\widehat a^{\mathrm{cum}}_m = \arg\max_{a\in\mathcal A} \widehat{\mu}_{m-1}(a),
\end{align*}
where $\widehat{\mu}_{m-1}(a)$ is the empirical mean yield of arm $a$ using only data in $\mathcal H_{m-1}$. The selected arm is then recommended for all decision units in year $m$. For the first year, for which no past data are available, a fixed fallback arm can be used.

This approach can be easily extended to the state cumulative yearly adaptive policy where for each state $q$, we select,
\begin{align*}
\widehat a^{\mathrm{cum}}_{q,m} = \arg\max_{a\in\mathcal A} \widehat{\mu}_{q,m-1}(a),
\end{align*}
where $\widehat{\mu}_{q,m-1}(a)$ is the empirical mean yield of arm $a$ in state $q$ using only data observed before year $m$. The site yearly adaptive policy is defined analogously at the site level. If a site has no previous observations, we fall back to the corresponding state-level choice. If a state has no previous observations, we fall back to the global cumulative choice.

\subsection{Sequential screening procedures (Arm elimination)}

The recommendation strategies considered in the previous subsection focus on selecting a single fertilizer choice. However, in practice multiple fertilizer choices often yield comparable outcomes despite their substantial differences in N application rates. A natural alternative is to identify a subset of near-optimal fertilizer choices rather than a unique best arm. 

\subsubsection{Batched State-Wise Screening}
\label{BSWS}

At the beginning of the first batch, every state starts with the full arm set, $\mathcal A_{q,1}=\mathcal A$. After observing yield outcomes from year $m$, we update empirical means and confidence intervals for each active arm within each state. Let $\widehat{\mu}_{q,m}(a)$ denote the empirical mean yield for arm $a$ in state $q$ using data observed up to and including batch $m$, and let
\begin{align*}
\mathrm{LCB}_{q,m}(a)= \widehat{\mu}_{q,m}(a) - z_{q,m}\widehat{\mathrm{se}}_{q,m}(a), \quad \text{and}, \quad \mathrm{UCB}_{q,m}(a)= \widehat{\mu}_{q,m}(a) + z_{q,m}\widehat{\mathrm{se}}_{q,m}(a)
\end{align*}
denote \emph{lower confidence bounds (LCB)} and \emph{upper confidence bounds (UCB)} with 
$z_{q,m} = \Phi^{-1} \left(1-\frac{\alpha}{2|\mathcal A_{q,m}|}\right)$, 
where $\Phi^{-1}$ is the standard Gaussian quantile function, $\alpha$ is a confidence parameter, and $|\mathcal A_{q,m}|$ is the number of currently active arms in state $q$ at batch $m$, and $\widehat{\mathrm{se}}_{q,m}(a) = \frac{\widehat{\sigma}_{q,m}(a)}{\sqrt{n_{q,m}(a)}}$ is the standard error.  The critical value includes a Bonferroni-type adjustment based on the number of currently active fertilizer choices, providing a conservative safeguard against premature elimination due to multiple comparisons. For a practical tolerance $\epsilon>0$, arm $a$ is retained for the next year if
\begin{align*}
\mathrm{UCB}_{q,m}(a) \ge \max_{a'\in\mathcal A_{q,m}} \mathrm{LCB}_{q,m}(a') - \epsilon.
\end{align*}
Thus,
\begin{align}
\label{eq: screening_rule}
\mathcal A_{q,m+1} = \left\{a\in\mathcal A_{q,m}: \mathrm{UCB}_{q,m}(a) \ge \max_{a'\in\mathcal A_{q,m}} \mathrm{LCB}_{q,m}(a') - \epsilon \right\}.
\end{align}
After the final batch, the lowest-N arm among the surviving state-level arms i.e.: 
\begin{align}
\widehat a^{\mathrm{screen}}_q = \arg\min_{a\in\mathcal A_{q,M+1}} N(a),
\label{eq: state_level_screeninga}
\end{align}
is recommended. Confidence bounds are incorporated to avoid prematurely eliminating fertilizer choices that appear inferior due to sampling variability. Arms are removed only when their upper confidence bound falls sufficiently below the best-supported alternatives.

\subsubsection{Batched State Screening with Site-Adjusted Prediction}

Yield can be modeled using the N application rate as the primary predictor, together with other covariates that influence crop performance, including environmental, soil, and management factors. Therefore, prediction-based screening may be useful for providing site-specific recommendations. Specifically, a site-adjusted regression model can be fitted at the end of each batch using only the data observed up to and including year $m$. Because the relationship between yield and N application is often nonlinear, the N rate can be included in the model as both a linear and a quadratic term to capture diminishing returns at higher application levels.  Let $Y_t$ denote grain yield, $a_t\in\mathcal A$ be the fertilizer arm, $X_t$ be pre-treatment covariates, and $\mathrm{State}(t)$ and $\mathrm{Site}(t)$ denote the corresponding state and site. The following quadratic model will be used to model the response
\begin{align*}
Y_t = \alpha_{\mathrm{Site}(t)} + \beta_1 N(a_t) + \beta_2 N(a_t)^2 + \gamma^\top X_t + \mathrm{State}(t):N(a_t) + \varepsilon_t.
\end{align*}
Using the model fit through batch $m$, let $\widehat Y_{t,m}(a)$ denote the predicted yield at decision unit $t$ if arm $a$ were assigned. We average these predicted yields within each state-arm pair:
\begin{align*}
\widehat m_{q,m}(a) = \frac{1}{n_{q,m}} \sum_{t:\mathrm{State}(t)=q,\;t\le m} \widehat Y_{t,m}(a).
\end{align*}
The predicted active set for the next batch is
\begin{align*}
\widehat{\mathcal A}^{\mathrm{pred}}_{q,m+1} = \left\{ a\in\mathcal A: \max_{a'\in\mathcal A}\widehat m_{q,m}(a') - \widehat m_{q,m}(a) \le \epsilon \right\}.
\end{align*}
After the final batch, the recommendation is the lowest-N arm among the final predicted near-best set, given by, 
\begin{align}
\widehat a^{\mathrm{pred}}_q = \arg\min_{a\in\widehat{\mathcal A}^{\mathrm{pred}}_{q,M+1}} N(a). 
\label{eq: state_level_screeningb}
\end{align}
Although this approach could potentially provide a more site-specific recommendation, the quality of the recommendation would depend heavily on the predictive performance of the model. In other words, if the model does not predict well, the resulting recommendation is also likely to be poor.

\subsection{Proposed Batched State-to-Site Hierarchical Refinement}
To combine stable state-level screening with local site-level refinement, we propose employing a hierarchical refinement strategy consisting of two stages.  The first stage is the batched state-wise screening procedure described in Section~\ref{BSWS}, which yields a final surviving set $\mathcal A_{q,M+1}$ for each state $q$. In the second stage, for each site $\ell$ within state $q$, we restrict attention to the state-level surviving arms and compute site-specific empirical means,
\begin{align*}
\widehat{\mu}_{q,\ell}(a) = \frac{1}{n_{q,\ell,a}} \sum_{t:\mathrm{State}(t)=q,\;\mathrm{Site}(t)=\ell,\;a\text{ observed at }t}Y_t(a), \qquad a\in\mathcal A_{q,M+1}.
\end{align*}
We then define the site-level near-best set
\begin{align*}
\mathcal A_{q,\ell}^{\mathrm{site}} = \left\{a\in\mathcal A_{q,M+1}: \max_{a'\in\mathcal A_{q,M+1}} \widehat{\mu}_{q,\ell}(a') - \widehat{\mu}_{q,\ell}(a) \le \epsilon_{\mathrm{site}} \right\}.
\end{align*}
The final recommendation for site $\ell$ in state $q$ is the lowest-N arm among the site-level near-best arms:
\begin{align*}
\widehat a_{q,\ell}^{\mathrm{hier}} = \arg\min_{a\in\mathcal A_{q,\ell}^{\mathrm{site}}} N(a).
\end{align*}
If a site has insufficient replication for site-level refinement, we fall back to the state-level low-N survivor from \eqref{eq: state_level_screeninga}.

\subsubsection{Theoretical screening guarantee for subset selection}

The proposed hierarchical refinement procedure begins with a state-level screening stage that removes fertilizer choices that appear substantially inferior based on historical observations. The following results show that, under standard large-sample conditions, the screening rule retains all truly near-optimal fertilizer choices with high probability. All proofs and discussion of the results are in the Appendix.

\begin{assumption}[State-level sampling]
\label{ass:state_sampling}

Fix a state $q$ and fertilizer choice $a$. Let $Y_{q,1}(a),\ldots,Y_{q,n_{q,m}(a)}(a)$ denote the observations for fertilizer choice $a$ available within state $q$ up to batch $m$. Assume that these observations are independent with common mean $\mathbb E[Y_{q,i}(a)] = \mu_q(a)$, and common finite variance $\operatorname{Var}(Y_{q,i}(a))=\sigma_q^2(a)<\infty$. Define the state-level empirical mean
\begin{align*}
\widehat{\mu}_{q,m}(a) = \frac{1}{n_{q,m}(a)}\sum_{i=1}^{n_{q,m}(a)}Y_{q,i}(a).
\end{align*}
\end{assumption}

\begin{proposition}[Central limit theorem]
\label{prop:state_clt}
Under Assumption~\ref{ass:state_sampling},
\begin{align*}
\sqrt{n_{q,m}(a)}\left(\widehat{\mu}_{q,m}(a) - \mu_q(a)\right)\xrightarrow{d}
N\left(0, \sigma_q^2(a)\right)
\end{align*}
as $n_{q,m}(a)\to\infty$.
\end{proposition}
In the implementation, the standard error $\widehat{\mathrm{se}}_{q,m}(a)$ is computed as the empirical standard deviation of the observed state--arm yields divided by $\sqrt{n_{q,m}(a)}$.

\begin{corollary}[Simultaneous confidence coverage]
\label{cor:coverage}
Let $\alpha > 0$ denote a confidence level. Define,
\begin{align*}
\mathrm{LCB}_{q,m}(a) = \widehat{\mu}_{q,m}(a) - z_{q,m} \widehat{\mathrm{se}}_{q,m}(a),\ \mathrm{UCB}_{q,m}(a) = \widehat{\mu}_{q,m}(a) + z_{q,m} \widehat{\mathrm{se}}_{q,m}(a),
\end{align*}
where for a Gaussian distribution, $\Phi$,
\begin{align*}
z_{q,m} = \Phi^{-1} \left(1 - \frac{\alpha} {2|\mathcal A_{q,m}|}\right).
\end{align*}
Suppose $\widehat{\mathrm{se}}_{q,m}(a)$ is a consistent estimator of the standard deviation of $\widehat{\mu}_{q,m}(a)$. Then, asymptotically,
\begin{align} 
\label{eq: uniform_coverage}
\Pr\Big(\mu_q(a) \in [\mathrm{LCB}_{q,m}(a), \mathrm{UCB}_{q,m}(a)] \text{ for all } a\in\mathcal A_{q,m} \Big) \ge 1-\alpha.
\end{align}
\end{corollary}

The screening rule \eqref{eq: screening_rule} retains any fertilizer choice whose upper confidence bound remains within $\epsilon$ bushels per acre of the largest lower confidence bound among the currently active choices. Intuitively, a fertilizer choice is eliminated only when the available evidence suggests that its mean yield is more than $\epsilon$ bushels per acre below that of a competing alternative. The following result shows that, with high probability, every truly $\epsilon$-near-optimal fertilizer choice survives the screening procedure.

\begin{proposition}[Screening safety]
\label{prop:screening_safety}
Define the true $\epsilon$-near-optimal set 
\begin{align}
\label{eq: true_optimal_set} 
\mathcal A_q^\star(\epsilon):= \left\{a\in\mathcal A_{q,m} : \mu_q^\star-\mu_q(a) \le \epsilon \right\},
\end{align}
where $\mu_q^\star = \max_{a\in\mathcal A_{q,m}} \mu_q(a)$. Recall, $\mathcal{A}_{q,m}$ is the set of all arms that survive the elimination round at the end of batch $m-1$. The screening rule,
\begin{align} 
\label{eq: screening_rule1}
\widehat{\mathcal A}_{q,m+1} := \left\{a\in\mathcal A_{q,m}:\mathrm{UCB}_{q,m}(a) \ge \max_{a’\in\mathcal A_{q,m}} \mathrm{LCB}_{q,m}(a’) - \epsilon \right\}
\end{align}
retains all truly $\epsilon$-near-optimal fertilizer choices with probability at least $1-\alpha$:
\begin{align*}
\Pr\Big(\mathcal A_q^\star(\epsilon)\subseteq\widehat{\mathcal A}_{q,m+1}\Big) \ge 1-\alpha.
\end{align*}
\end{proposition}
Proposition~\ref{prop:screening_safety} provides a screening-safety guarantee for the state-level elimination step. Specifically, with probability at least $1-\alpha$, every fertilizer choice whose mean yield lies within $\epsilon$ bushels per acre of the optimal state-level yield survives the screening procedure. Thus, the screening rule removes only fertilizer choices that are sufficiently separated from the optimum while preserving all practically competitive alternatives.

In our proposed method, the retained set at the state level is subsequently used within the hierarchical refinement procedure. At the state level, the screening rule removes fertilizer choices that appear substantially inferior while preserving those that remain statistically competitive. In the final refinement stage, site-specific information is used to select among the surviving state-level candidates, allowing recommendations to adapt to local growing conditions while maintaining the stability provided by state-level information sharing.

\subsection{Policy evaluation}
\label{sec: policy_evaluation}
All recommendation procedures are evaluated retrospectively using RCBD. Because multiple fertilizer choices were observed within each trial-block decision unit, the data provide a natural benchmark for assessing the quality of a recommendation. Specifically, the yield obtained by a recommended fertilizer choice can be compared with the best observed fertilizer choice within the same trial-block unit. Let $\pi:\mathcal X \rightarrow \mathcal A$ denote a recommendation procedure, where $\mathcal X$ denotes the covariate space and $\mathcal A$ denotes the set of fertilizer choices. Given covariates $X_t$, the procedure recommends a fertilizer choice $\widehat a_t=\pi(X_t)$ and receives the corresponding observed yield $Y_t(\widehat a_t)$. For a decision unit $t$, let $a_t^{\star} = \arg\max_{a\in\mathcal A} Y_t(a)$ denote the empirically best observed arm for decision unit $t$. If a recommendation procedure $\pi$ selects fertilizer choice $\widehat a_t$, the instantaneous \emph{regret} at decision unit $t$ is defined as
\begin{align*}
R_t (\pi_t = \hat{a}_t) = Y_t(a_t^\star) - Y_t(\widehat a_t).
\end{align*}
Thus, regret measures the yield loss associated with a recommended fertilizer choice relative to the best fertilizer choice observed within the same trial-block experimental unit. Smaller regret therefore corresponds to better recommendation performance. For a policy evaluated across $T$ decision units, we summarize performance using the \emph{cumulative regret},
\begin{align} \label{eq: regret_definition}
\mathcal{R}(T) = \sum_{t=1}^T R_t (\pi_t).
\end{align}
The \emph{mean regret} then is $\mathcal{R}(T)/T$.  

While regret provides a useful measure of yield loss, it does not account for differences in N application rates. Because the exploratory analyses suggest that several fertilizer choices often achieve similar yields, recommendations with slightly larger regret may still be preferable if they substantially reduce N use. Therefore, we also evaluate the ability of each procedure to identify a set of practically near-optimal fertilizer treatments. For a user-specified tolerance level $\epsilon > 0$, define the empirical \emph{best-subset} at decision unit $t$ as
\begin{align*}
\mathcal A_t^\star(\epsilon) = \left\{a\in\mathcal A: Y_t(a_t^\star)-Y_t(a)\le\epsilon \right\},
\end{align*}
Thus, $\mathcal A_t^\star(\epsilon)$ contains all fertilizer choices whose observed yield is within $\epsilon$ bushels per acre of the best observed choice for that decision unit. A recommendation is considered successful under the best-subset criterion if $\widehat a_t \in \mathcal A_t^\star(\epsilon)$. 
We report the proportion of decision units for which the selected fertilizer choice belongs to the empirical best subset for $\epsilon=5$ and $\epsilon=10$ bushels per acre. These measures quantify the frequency with which a procedure identifies a practically near-optimal fertilizer choice, even when the selected choice is not the unique highest-yielding alternative.

In addition to regret-based measures, we report average grain yield, average total N application, and the proportion of recommendations that fall within 5 and 10 bushels per acre of the empirically best fertilizer choice. Together, these metrics quantify the tradeoff between yield performance and N use.

\section{Real Data Analysis}
\label{real}

\subsection{Global, State, and Site-Level Heterogeneity}\label{sec: exploratory}
The analyses in this section are exploratory and are conducted on the full dataset. Consequently, the results should not be interpreted as representing deployable recommendation procedures. Rather, the objective is to characterize the structure of the nitrogen response surface and assess the degree of treatment heterogeneity present in the data.
To assess the extent of spatial heterogeneity, we compared following regression models: 
\begin{enumerate}
    \item \textbf{Model 1:} a globally linear N response surface.
    \item \textbf{Model 2:} Add a quadratic N terms to Model 1 to account for nonlinear behavior.
    \item \textbf{Model 3:} A Mixed Effect model that allows N effects to vary by state through state-specific interaction terms.
    \item \textbf{Model 4:} Add a site-specific fixed effects to capture local productivity differences.
\end{enumerate}
Table~\ref{tab:model-structure} shows the components of each model. Here, $X_t$ includes pre-treatment covariates such as soil nitrate measurements and site productivity indicators, while $\alpha_{\mathrm{Site}(t)}$ denotes a site-specific intercept. 
\begin{table}[h]
\centering
\caption{Regression models used in the exploratory retrospective analysis.}
\label{tab:model-structure}
\begin{tabular}{lcccc}
\hline
Component & Model 1 & Model 2 & Model 3 & Model 4 \\
\hline
Linear N term $N(a_t)$& \checkmark & \checkmark & \checkmark & \checkmark \\
Quadratic N term $N(a_t)^2$ &  & \checkmark & \checkmark & \checkmark \\
Pre-treatment covariates $X_t$ & \checkmark & \checkmark & \checkmark & \checkmark \\
State fixed effects &  &  & \checkmark & \checkmark \\
State $\times$ N interactions &  &  & \checkmark & \checkmark \\
Site fixed effects &  &  &  & \checkmark \\
\hline
\end{tabular}
\end{table}

As it can be seen in Table \ref{tab:model-comparison}, the results indicate several important features of the dataset. First, introducing quadratic N terms substantially improves predictive performance, confirming the nonlinear N response suggested by Figure~\ref{fig:yield-total-n}. Second, allowing N effects to vary by state further improves model fit, indicating meaningful regional heterogeneity in treatment response. Finally, incorporating site-level fixed effects leads to a substantial increase in explained variation, with the final model achieving approximately $R^2\approx 0.69$. This suggests that local site-level conditions contribute strongly to yield variability and that globally pooled fertilizer recommendations may be insufficient.

\begin{table}
\centering
\caption{Model comparison for retrospective yield prediction.}
\label{tab:model-comparison}
\begin{tabular}{lcccc}
  \hline
 Model & $R^2$ & Adjusted $R^2$ & AIC & BIC \\
  \hline
Model 1 & 0.39 & 0.39 & 30611.00 & 30689.29 \\ 
Model 2 & 0.47 & 0.47 & 30191.77 & 30276.08 \\ 
Model 3 & 0.50 & 0.49 & 30031.53 & 30176.06 \\ 
Model 4 & 0.69 & 0.68 & 28634.65 & 28911.67 \\ 
   \hline
\end{tabular}
\end{table}

\subsection{Comparison of Fertilizer Recommendation Strategies}

The primary goal of this study is to assess whether location based adaptive recommendation strategies can improve yield performance while simultaneously reducing N application relative to globally pooled recommendations. We compare the fertilizer recommendations using the retrospective evaluation framework described in Section~\ref{sec: policy_evaluation}. Table~\ref{tab:deployable-policy-comparison} summarizes the performance of the deployable fertilizer recommendation procedures. 

These results indicate that the state-level screening procedure performs poorly, achieving the lowest average yield and the highest average regret among all methods. Although it substantially reduces N application, its screening rule appears overly aggressive under limited yearly replication, leading to the premature elimination of many fertilizer choices. This finding highlights the need to balance N reduction objectives with adequate protection against early elimination errors. The results also show that allowing recommendations to vary by state improves performance relative to a single global recommendation, reducing average regret from 18.3 to 16.2 bushels per acre and increasing average yield from 212.7 to 214.8 bushels per acre, which indicates meaningful regional heterogeneity in N response. By contrast, the cumulative yearly adaptive procedures provide little improvement over their fixed counterparts: the adaptive global recommendation performs identically to the global recommendation, whereas the adaptive state and adaptive site procedures exhibit slightly higher regret. This limited gain is likely due to the small number of years available for sequential updating, which restricts the amount of information that can be accumulated before recommendations are revised. The relatively weak performance of the site-adjusted prediction screening procedure (Method 6) is noteworthy because its underlying regression model had the strongest retrospective predictive performance among the models considered in the exploratory analysis in Section~\ref{sec: exploratory}. This suggests that accurate prediction of yield levels alone may be insufficient for identifying near-optimal fertilizer recommendations.

The proposed hierarchical refinement procedure achieves the strongest overall performance among the deployable methods, with the highest average yield (217.4 bushels per acre), lowest average regret (13.6 bushels per acre), and highest proportion of near-optimal recommendations (49\%). Notably, it also substantially reduces average N application, from approximately 240--247 units under the global and state-specific recommendations to 180 units. Relative to the state-specific recommendation, it increases average yield by 2.6 bushels per acre, reduces regret by 2.6 bushels per acre, and lowers N application by 67.0 units. These results suggest that recommendation sets based on near-optimal fertilizer choices can reduce N use without sacrificing yield. The stronger performance of hierarchical refinement suggests that combining state-level screening with local refinement provides a more robust strategy for identifying practically near-optimal fertilizer choices.

\begin{table}
\centering
\caption{Retrospective comparison of deployable fertilizer recommendation procedures. All methods are evaluated on the same 185 trial-block decision units.} 
\label{tab:deployable-policy-comparison}
\scalebox{0.85}{\begin{tabular}{p{6.15cm}rrrr}
  \hline
Policy & Mean yield & Mean total N & Mean regret & Within 10 bu/ac\\ 
  \hline
Global recommendation & 212.68 & 240.00 & 18.32 & 0.34 \\ 
State recommendation & 214.76 & 246.70 & 16.24 & 0.42 \\ 
Adaptive global & 212.68 & 240.00 & 18.32 & 0.34 \\ 
Adaptive state & 211.20 & 239.78 & 19.81 & 0.31 \\ 
Adaptive site & 211.50 & 236.32 & 19.51 & 0.28 \\ 
State screening & 178.85 & 70.49 & 52.16 & 0.04 \\ 
Site-adjusted prediction screening & 207.21 & 185.30 & 23.80 & 0.30 \\ 
Hierarchical refinement (Proposed) & 217.43 & 179.68 & 13.58 & 0.49 \\ 
   \hline
\end{tabular}}
\end{table}

Table~\ref{tab:state-summary-proposed} summarizes the fertilizer recommendations produced by the proposed hierarchical refinement procedure. It can be seen that no state is associated with a single universally recommended fertilizer choice. Instead, between two and five distinct fertilizer choices are recommended within each state, reflecting substantial within-state heterogeneity in growing conditions and treatment response. Furthermore, the most frequently recommended fertilizer choice varies across states, with dominant total N rates ranging from 80 to 200 units. Finally, the dominant recommendation typically accounts for only 33\%--57\% of decision units within a state, indicating that a substantial fraction of sites benefit from recommendations that differ from the most common state-level choice. These findings further support the conclusion that globally uniform fertilizer recommendations may be inadequate and that local adaptation can be beneficial.  More details and comparisons of recommendations at different sites within the states are provided in the Appendix. 
\begin{table}
\centering
\caption{Summary of state-level fertilizer recommendation under the proposed hierarchical refinement procedure. NRC denotes the number of distinct fertilizer choices assigned within a state; the most common choice is the planting plus sidedress N combination assigned to the largest proportion of trial-block decision units.} 
\label{tab:state-summary-proposed}
\scalebox{0.85}{\begin{tabular}{lrlrr}
  \hline
State & NRC & Most common choice & Total N & Percent of units \\ 
  \hline
IA &   4 & 160+0 & 160 & 50.00 \\ 
  IL &   3 & 40+160 & 200 & 50.00 \\ 
  IN &   2 & 80+80 & 160 & 50.00 \\ 
  MN &   4 & 80+80 & 160 & 34.80 \\ 
  MO &   4 & 200+0 & 200 & 57.10 \\ 
  ND &   2 & 40+80 & 120 & 50.00 \\ 
  NE &   3 & 80+0 &  80 & 50.00 \\ 
  WI &   5 & 40+160 & 200 & 33.30 \\ 
   \hline
\end{tabular}}
\end{table}

\subsection{Hindsight Benchmarks Analysis}
Table~\ref{tab:hindsight-policy-comparison} reports hindsight reference procedures that use information from the same year being evaluated and therefore cannot be deployed in practice. These procedures provide optimistic benchmarks for the best performance achievable under global, state-level, and site-level adaptation. As expected, increasingly localized hindsight recommendations improve performance, with average regret decreasing from 18.0 bushels per acre for the global-year benchmark to 8.9 bushels per acre for the site-year benchmark. This pattern provides additional evidence that substantial location-specific variation exists in fertilizer response. Although the proposed hierarchical refinement procedure does not attain the performance of the site-year hindsight benchmark, it closes a substantial portion of the gap between globally pooled recommendations and the optimistic oracle benchmark while using approximately 20\% less N than the global recommendation. This suggests that much of the benefit of site-specific adaptation can be captured through a practical hierarchical refinement strategy without requiring perfect knowledge of same-year outcomes.

\begin{table}
\centering
\caption{Hindsight reference procedures using outcomes from the same year being evaluated. These procedures are not deployable and serve only as optimistic benchmarks.} 
\label{tab:hindsight-policy-comparison}
\scalebox{0.85}{\begin{tabular}{p{4cm}rrrr}
  \hline
Policy & Mean yield & Mean total N & Mean regret & Within 10 bu/ac \\ 
  \hline
Global-year hindsight & 213.05 & 266.81 & 17.96 & 0.36 \\ 
State-year hindsight & 218.66 & 220.97 & 12.35 & 0.54 \\ 
Site-year hindsight & 222.13 & 223.78 & 8.88 & 0.66 \\ 
   \hline
\end{tabular}}
\end{table}

\subsection{Sensitivity Analyses and Out-of-Sample Evaluation}
The analyses presented above evaluate recommendation procedures retrospectively using all available trial-block observations. While this framework is useful for comparing competing recommendation strategies, it does not fully assess how well recommendations generalize to previously unseen experimental units. We conducted an out-of-sample evaluation based on held-out trial blocks and a sensitivity analysis examining the effect of the near-optimality threshold used in the proposed sequential screening and hierarchical refinement procedure.

\subsubsection{Held-Out Block Evaluation}

Most site-year trials in the dataset contain four randomized blocks. To assess out-of-sample performance, we performed a leave-one-block-out evaluation. For each fold, one block from every trial was designated as a test set, while the remaining blocks were used to construct all recommendation procedures. Recommendations were then evaluated on the held-out block only. This design preserves the spatial and temporal structure of the experiment while ensuring that the outcomes used for evaluation were not used to construct the recommendations. Table~\ref{tab:heldout-policy-comparison} summarizes the result of this analysis. 

\begin{table}[h!]
\centering
\caption{Held-out block evaluation of fertilizer recommendation procedures. For each fold, one randomized block from each site-year trial was held out for evaluation and the remaining blocks were used to construct recommendations.} 
\label{tab:heldout-policy-comparison}
\scalebox{0.8}{\begin{tabular}{p{6.15cm}rrrr}
  \hline
Policy & Mean yield & Mean total N & Mean regret & Within 10 bu/ac \\ 
  \hline
Global recommendation & 211.79 & 249.74 & 18.97 & 0.32 \\ 
State recommendation & 211.61 & 236.88 & 19.24 & 0.31 \\ 
Adaptive global & 211.79 & 243.32 & 18.80 & 0.33 \\ 
Adaptive state & 209.72 & 241.86 & 20.94 & 0.30 \\ 
Adaptive site & 209.55 & 240.21 & 21.10 & 0.29 \\ 
State screening & 178.67 & 71.55 & 51.89 & 0.05 \\ 
Site-adjusted prediction screening & 206.03 & 180.31 & 24.63 & 0.27 \\ 
Hierarchical refinement (Proposed) & 211.82 & 191.55 & 18.64 & 0.32 \\ 
  \hline
Global-year hindsight & 210.09 & 243.32 & 20.67 & 0.32 \\ 
State-year hindsight & 212.63 & 221.35 & 18.08 & 0.38 \\ 
Site-year hindsight & 211.20 & 222.38 & 19.52 & 0.37 \\ 
   \hline
\end{tabular}}
\end{table}
The overall ranking of methods remains largely unchanged relative to the retrospective analysis. The proposed hierarchical refinement procedure continues to achieve the smallest regret among the deployable recommendation methods while maintaining substantially lower N application rates than the global and state recommendations. In particular, the proposed procedure achieves an average regret of 18.64 bushels per acre while recommending approximately 192 units of N on average, compared with 237-250 units for the global and state recommendation policies. Overall, these results suggest that the benefits of the hierarchical refinement strategy are not solely a consequence of retrospective evaluation and persist under out-of-sample assessment.

The site-adjusted prediction screening procedure again performs noticeably worse than the proposed method, despite relying on the strongest predictive model identified in the exploratory analysis. Therefore, the leave-one-block-out results also provide additional evidence that recommendation quality depends on accurately distinguishing among treatments with very similar expected yields rather than simply maximizing predictive accuracy. Methods based on direct identification of near-optimal fertilizer choices appear to generalize more effectively than approaches that rely primarily on predicted yield levels.   These results highlight the complementary roles of sequential screening for safely eliminating inferior treatments and hierarchical refinement for adapting recommendations to local growing conditions.

\subsubsection{Sensitivity to Near-Optimality Threshold}

The proposed hierarchical refinement procedure uses a threshold parameter $\epsilon$ to define the set of fertilizer choices considered practically near-optimal at the state-level screening stage. Larger values of $\epsilon$ retain more fertilizer choices, potentially reducing N usage by allowing lower-input treatments to survive the screening stage, while smaller values produce more aggressive elimination. To examine the robustness of the proposed procedure, we repeated the held-out block evaluation over a range of values of $\epsilon$. Figure~\ref{fig:epsilon-sensitivity} summarizes the resulting trade-off between mean regret and average N application.

The sensitivity analysis indicates that the choice of $\epsilon$ controls the trade-off between N reduction and yield performance. Smaller values of $\epsilon$ produce more conservative recommendations with lower regret but higher N use, whereas larger values of $\epsilon$ favor lower-input fertilizer programs at the expense of increased regret. Figure~\ref{fig:epsilon-sensitivity} suggests the presence of an elbow around $\epsilon \approx 10$. Up to this point, increasing $\epsilon$ yields substantial reductions in N application with relatively modest increases in regret. Beyond $\epsilon \approx 10$, further reductions in N use are accompanied by a more rapid deterioration in regret performance. This pattern suggests that values of $\epsilon$ near/below 10 may provide a particularly attractive compromise between agronomic performance and fertilizer reduction. These results suggest that the qualitative conclusions of the study are robust to the specific threshold choice used in the main analysis. 

\begin{figure}
    \centering
    \includegraphics[width=0.5\linewidth]{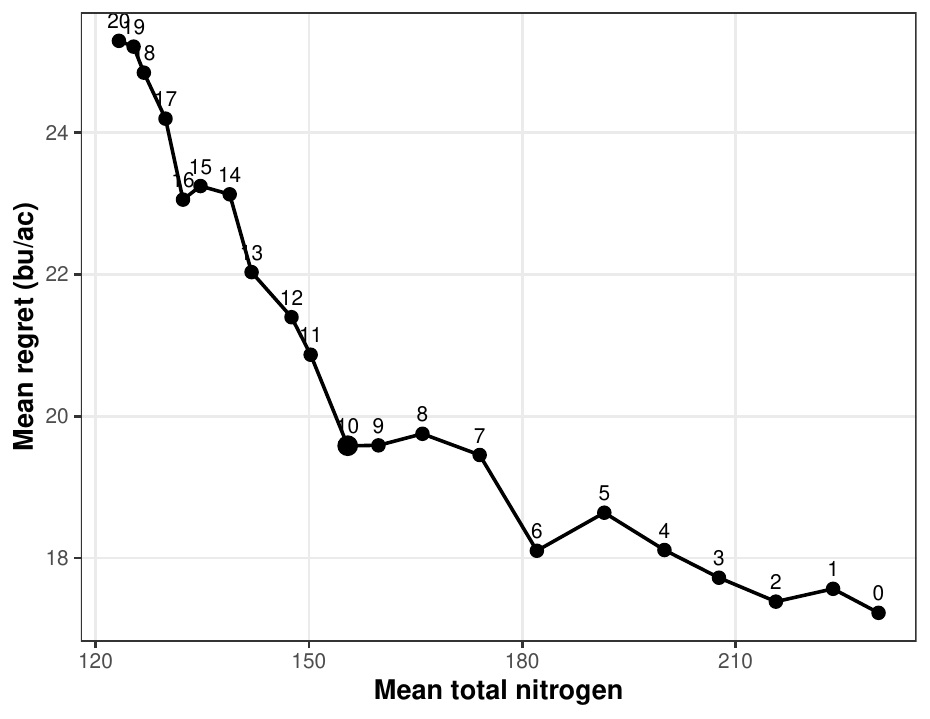}
  \caption{Held-out block yield--N tradeoff for the proposed sequential screening and hierarchical refinement procedure across values of the near-optimality threshold $\epsilon$. Points are labeled by $\epsilon$. Larger values of $\epsilon$ allow lower-N recommendations but eventually lead to larger increases in regret.}
    \label{fig:epsilon-sensitivity}
\end{figure}

Overall, the real-data analysis provides consistent evidence for three key conclusions. First, N response is nonlinear and exhibits broad regions of near-optimal performance, implying that multiple fertilizer choices can achieve similar yields. Second, substantial heterogeneity exists across states and sites, suggesting that globally uniform fertilizer recommendations are unlikely to be optimal. Third, recommendation strategies that explicitly identify sets of near-optimal fertilizer choices and subsequently favor lower-N alternatives can achieve favorable yield--N tradeoffs. In particular, the proposed sequential screening and hierarchical refinement procedure simultaneously improves average yield, reduces regret, and lowers N application relative to both global and state-level recommendation strategies. These findings suggest that hierarchical refinement may provide a practical framework for developing N recommendations that balance productivity and fertilizer efficiency in heterogeneous agricultural environments.

\section{Discussion and Conclusion}
\label{sec: discussion}
In this paper, we develop a sequential screening and hierarchical refinement framework for identifying near-optimal N fertilizer recommendations in multi-site agricultural experiments. Instead of seeking a single globally optimal fertilizer treatment, the proposed approach recognizes that N response surfaces are often relatively flat near the optimum. Therefore, multiple fertilizer choices can achieve comparable yields while requiring substantially different N application rates. The real-data analysis highlights that fertilizer response varies meaningfully across states and sites. This indicates that globally uniform N recommendations may be too coarse for heterogeneous production environments. Furthermore, purely yearly adaptive best-arm strategies provide limited improvement in this short three-year study. This suggests that the amount of historical information available before each update is insufficient for stable fully adaptive learning. Moreover, the proposed state-to-site hierarchical refinement procedure achieves the most favorable yield–N tradeoff among the deployable methods considered. In particular, it improves average yield and reduces regret relative to both global and state-level recommendations while substantially lowering average N application. Finally, the results highlight an important distinction between prediction and decision-making. The site-adjusted model achieved the strongest retrospective predictive performance among the models considered in the exploratory analysis, yet the corresponding prediction-based screening procedure did not yield the best fertilizer recommendations. This suggests that accurate prediction of yield levels alone may be insufficient for recommendation tasks, where the primary objective is to distinguish among treatments with similar expected yields and identify practically near-optimal alternatives. It should be noted that important drivers of N response such as weather conditions, precipitation patterns, temperature, and other year-specific environmental factors, were not explicitly incorporated into the predictive model. As a result, model-based predicted differences among fertilizer choices may contain substantial uncertainty and may not reliably distinguish among treatments whose expected yields are very similar.

The proposed hierarchical refinement procedure addresses this challenge by combining state-level sequential screening with site-level refinement, rather than relying on model-based predictions. State-level information is first used to screen out clearly inferior fertilizer choices, after which site-level recommendations are based on empirical mean yields among the surviving alternatives. By restricting local comparisons to a smaller set of plausible fertilizer choices, the procedure avoids the need to accurately model yield differences across all treatments and may provide a more robust balance between information sharing and local adaptation. The practical importance of these findings extends beyond environmental considerations. As N fertilizer prices become increasingly volatile, recommendation strategies that maintain yield while reducing fertilizer requirements may provide substantial economic benefits to producers. By identifying near-optimal fertilizer choices rather than a single yield-maximizing treatment, the proposed framework offers a mechanism for improving N-use efficiency while potentially reducing input costs. Overall, the results suggest that sequential screening of near-optimal fertilizer choices provides a useful framework for precision N management. By shifting the objective from identifying a unique best treatment to identifying a set of practically equivalent treatments, the proposed approach enables lower-N recommendations to be selected when yield differences are small. This perspective may be broadly useful in agricultural decision-making settings where productivity, economic costs, and environmental impacts must be balanced under substantial spatial heterogeneity.

Several limitations should be noted. First, the study uses data from 49 site-years over only three growing seasons. Therefore the amount of temporal information for sequential updating is limited. Second, important drivers of nitrogen response, including weather, precipitation, soil characteristics, and other year specific environmental variables, were not explicitly incorporated into the recommendation procedure. Integrating such information through more sophisticated models may further improve future site-specific refinement recommendation procedures. Third, the screening criterion was based solely on yield and nitrogen application rates. Incorporating fertilizer costs and economic returns directly into the screening procedure would allow the recommendations to be based on economic returns in addition to agronomic performance. These extensions represent promising directions for future research.

\section{Disclosure statement}\label{disclosure-statement}
The authors declare no relevant disclosures of competing interests.

\section{Acknowledgments}
The authors would like to thank Jason Clark for providing valuable feedback on the agronomic validity of the results.

 \bibliographystyle{abbrvnat}

\bibliography{Reference.bib}

\newpage

\appendix
\section{Theoretical and Additional Data Analysis Results}
\subsection{Theoretical Result}
\subsubsection{Proof of Proposition \ref{prop:state_clt}}
\begin{proof}
The result follows directly from the classical central limit theorem.
\end{proof}

\subsubsection{Proof of Corollary \ref{cor:coverage}}

\begin{proof}
By Proposition~\ref{prop:state_clt}, consistency of $\widehat{\mathrm{se}}_{q,m}(a)$, and Slutsky’s theorem,
\begin{align*}
\frac{\widehat{\mu}_{q,m}(a)-\mu_q(a)}{\widehat{\mathrm{se}}_{q,m}(a)} \xrightarrow{d} N(0,1).
\end{align*}
Therefore,
\begin{align*}
\Pr\Big(\mu_q(a)\notin[\mathrm{LCB}_{q,m}(a), \mathrm{UCB}_{q,m}(a)] \Big) \le \frac{\alpha}{|\mathcal A_{q,m}|}
\end{align*}
for each active arm. Applying the union bound,
\begin{align*}
\Pr\left(\exists a\in\mathcal A_{q,m} : \mu_q(a) \notin [\mathrm{LCB}_{q,m}(a), \mathrm{UCB}_{q,m}(a)] \right) &\le \sum_{a\in\mathcal A_{q,m}} \frac{\alpha}{|\mathcal A_{q,m}|} \\
&= \alpha.
\end{align*}
Taking complements,
\begin{align*}
\Pr&\left( \mu_q(a) \in [\mathrm{LCB}_{q,m}(a), \mathrm{UCB}_{q,m}(a)] \text{ for all } a\in\mathcal A_{q,m} \right)\\
&= 1-\Pr\left(\exists a\in\mathcal A_{q,m} : \mu_q(a) \notin [\mathrm{LCB}_{q,m}(a), \mathrm{UCB}_{q,m}(a)] \right) \\
&\ge 1-\alpha.
\end{align*}
\end{proof}

The coverage result in Corollary~\ref{cor:coverage} relies on an
asymptotic normal approximation for the state-level empirical means.
Although the study spans only three growing seasons, each state--arm
combination typically contains a moderate number of observations
accumulated across multiple sites and blocks. Consequently, the
state-level empirical means are based on substantially more observations
than the number of annual updates alone might suggest, making the normal
approximation a reasonable working assumption.

More generally, the screening guarantee does not depend specifically on
the use of asymptotic confidence intervals. Similar coverage and
screening-safety results could be established using finite-sample
concentration inequalities and corresponding confidence bounds.
We adopt the asymptotic formulation here because it closely matches the
confidence intervals used in the implementation and provides a simple
theoretical justification for the screening procedure.

\subsubsection{Proof of Proposition \ref{prop:screening_safety}}
\begin{proof}
By Corollary~\ref{cor:coverage}, with probability at least $1-\alpha$,
\begin{align*}
\mu_q(a) \in [\mathrm{LCB}_{q,m}(a),\mathrm{UCB}_{q,m}(a)]
\end{align*}
for all $a\in\mathcal A_{q,m}$. Consequently, by definition,
\begin{align}
\label{eq: step1a}
\max_{a^\prime \in\mathcal A_{q,m}} \mathrm{LCB}_{q,m}(a^\prime) \le \mu_q^\star, 
\end{align}
and
\begin{align}
\mathrm{UCB}_{q,m}(a) \ge \mu_q(a).
\label{eq:step1b}
\end{align}
Now let $a\in\mathcal A_q^\star(\epsilon)$. Then, from \eqref{eq: true_optimal_set},
\begin{align}
\label{eq: step2}
\mu_q(a) \ge \mu_q^\star-\epsilon.
\end{align}
It follows from \eqref{eq: step1a}, \eqref{eq:step1b} and \eqref{eq: step2} that,
\begin{align*}
\mathrm{UCB}_{q,m}(a) \ge \mu_q^\star-\epsilon \ge \max_{a^\prime\in\mathcal A_{q,m}} \mathrm{LCB}_{q,m}(a^\prime) -\epsilon.
\end{align*}
Note that $a$ satisfies \eqref{eq: screening_rule}, hence,
\begin{align*}
a\in\widehat{\mathcal A}_{q,m+1}.
\end{align*}
Since $a$ was arbitrary,
\begin{align*}
\mathcal A_q^\star(\epsilon) \subseteq \widehat{\mathcal A}_{q,m+1}.
\end{align*}
Combining this inclusion with Corollary~\ref{cor:coverage} yields the result.
\end{proof}



\subsection{Additional Data Analysis Results}
Table~\ref{tab:state-treatment-frequency} provides the full distribution of fertilizer recommendations under the proposed sequential screening and hierarchical refinement procedure. The results reveal that recommendations are concentrated on a relatively small subset of the available fertilizer choices, despite the presence of 16 candidate treatments. In particular, N programs corresponding to total N applications between approximately 120 and 200 units are selected most frequently across states, whereas the highest N treatments are rarely recommended. This pattern is consistent with the exploratory analyses, which indicated a relatively flat near-optimal N response surface and suggested that lower-N treatments can often achieve yields comparable to substantially higher N applications.
\begin{table}[ht]
\centering
\caption{State-specific distribution of fertilizer recommendations under the proposed hierarchical refinement procedure. Entries denote the percentage of trial-block decision units within each state assigned to each fertilizer choice.} 
\label{tab:state-treatment-frequency}
\begingroup\scriptsize
\begin{tabular}{lrrrrrrrrrrrrr}
  \hline
State & 40--40 & 80--0 & 40--80 & 120--0 & 40--120 & 80--80 & 160--0 & 40--160 & 200--0 & 40--200 & 80--160 & 240--0 & 40--240 \\ 
  \hline
IA & 0.00 & 0.00 & 0.00 & 16.70 & 0.00 & 0.00 & 50.00 & 16.70 & 0.00 & 0.00 & 0.00 & 0.00 & 16.70 \\ 
  IL & 0.00 & 0.00 & 0.00 & 0.00 & 0.00 & 0.00 & 16.70 & 50.00 & 33.30 & 0.00 & 0.00 & 0.00 & 0.00 \\ 
  IN & 0.00 & 0.00 & 0.00 & 0.00 & 0.00 & 50.00 & 0.00 & 50.00 & 0.00 & 0.00 & 0.00 & 0.00 & 0.00 \\ 
  MN & 0.00 & 0.00 & 0.00 & 0.00 & 0.00 & 34.80 & 30.40 & 0.00 & 0.00 & 17.40 & 0.00 & 17.40 & 0.00 \\ 
  MO & 0.00 & 0.00 & 0.00 & 0.00 & 0.00 & 0.00 & 0.00 & 14.30 & 57.10 & 14.30 & 14.30 & 0.00 & 0.00 \\ 
  ND & 0.00 & 0.00 & 50.00 & 0.00 & 0.00 & 0.00 & 0.00 & 0.00 & 50.00 & 0.00 & 0.00 & 0.00 & 0.00 \\ 
  NE & 0.00 & 50.00 & 0.00 & 0.00 & 33.30 & 0.00 & 0.00 & 0.00 & 16.70 & 0.00 & 0.00 & 0.00 & 0.00 \\ 
  WI & 16.70 & 0.00 & 0.00 & 0.00 & 0.00 & 0.00 & 16.70 & 33.30 & 0.00 & 0.00 & 16.70 & 16.70 & 0.00 \\ 
   \hline
\end{tabular}
\endgroup
\end{table}
Table~\ref{tab:example-site-recommendations} provides illustrative fertilizer recommendations for two randomly selected sites within each state.  The proposed hierarchical refinement procedure frequently recommends lower total N rates than the corresponding state-level recommendation. For example, in Iowa, Minnesota, Nebraska, and Wisconsin, the proposed method often selects recommendations with total N applications that are 40--120 units lower than the state-level recommendation. Second, the selected recommendations vary substantially across sites within the same state, reflecting the site-level refinement stage of the procedure. Finally, although the site-year hindsight benchmark occasionally recommends higher N rates, the proposed method often identifies fertilizer choices with substantially lower N application while remaining competitive in terms of yield performance. These examples illustrate how the hierarchical refinement strategy balances local adaptation with N reduction by selecting lower-N alternatives from a set of near-optimal fertilizer choices.

\begin{table}[ht]
\centering
\caption{Illustrative site-level fertilizer recommendations for two randomly selected sites within each state. Entries show planting and sidedress N rates, with total N in brackets.} 
\label{tab:example-site-recommendations}
\begingroup\scriptsize
\begin{tabular}{llp{3.2cm}p{3.2cm}p{3.2cm}}
  \hline
State & Site & Hierarchical refinement (Method 7) & Site-year hindsight (Method 3c) & State recommendation (Method 2) \\ 
  \hline
IA & Boone & (160, 0) [160] & (40, 200) [240] & (40, 200) [240] \\ 
  IA & Story & (160, 0) [160] & (200, 0) [200] & (40, 200) [240] \\ 
  IL & Shumway & (160, 0) [160] & (160, 0) [160] & (240, 0) [240] \\ 
  IL & Urbana & (40, 160) [200] & (240, 0) [240] & (240, 0) [240] \\ 
  IN & Loam & (80, 80) [160] & (40, 240) [280] & (40, 240) [280] \\ 
  IN & Sand & (40, 160) [200] & (80, 160) [240] & (40, 240) [280] \\ 
  MN & NewRichland & (160, 0) [160] & (280, 0) [280] & (80, 160) [240] \\ 
  MN & StCharles & (80, 80) [160] & (200, 0) [200] & (80, 160) [240] \\ 
  MO & Bay & (80, 160) [240] & (80, 160) [240] & (40, 240) [280] \\ 
  MO & Troth & (200, 0) [200] & (200, 0) [200] & (40, 240) [280] \\ 
  ND & Amenia & (200, 0) [200] & (80, 160) [240] & (200, 0) [200] \\ 
  ND & Durbin & (40, 80) [120] & (160, 0) [160] & (200, 0) [200] \\ 
  NE & Brandes & (40, 120) [160] & (40, 160) [200] & (40, 240) [280] \\ 
  NE & Kyes & (200, 0) [200] & (40, 200) [240] & (40, 240) [280] \\ 
  WI & Belmont & (40, 40) [80] & (40, 120) [160] & (40, 160) [200] \\ 
  WI & Steuben & (80, 160) [240] & (280, 0) [280] & (40, 160) [200] \\ 
   \hline
\end{tabular}
\endgroup
\end{table}
Table~\ref{tab:screening-path-sites} illustrates the screening process underlying the proposed hierarchical refinement procedure for a collection of representative sites. All sites begin with the full set of 16 fertilizer choices. State-level screening substantially reduces the candidate set while retaining multiple near-optimal alternatives, typically leaving between 12 and 15 surviving fertilizer choices. Subsequent site-level refinement further narrows the candidate set using site-specific information, often reducing the number of surviving choices to only a few alternatives. The final recommendation is then selected as the lowest-N fertilizer choice among the site-level survivors. These examples illustrate how the proposed procedure combines broad state-level information sharing with local site-level adaptation, allowing N application to be reduced while retaining fertilizer choices that remain competitive in terms of expected yield.
\begin{table}[ht]
\centering
\caption{Illustrative screening paths under the proposed hierarchical refinement procedure. All sites begin with the full set of 16 fertilizer choices. The table reports the number of state- and site-level survivors, the final site-level survivor set, and the resulting low-N recommendation.} 
\label{tab:screening-path-sites}
\begingroup\scriptsize
\begin{tabular}{llllp{4.0cm}p{3.0cm}}
  \hline
State & Site & State screening & Site screening & Site-level survivor set & Final recommendation \\ 
  \hline
IA & Boone & 14 survivors & 4 survivors & (160, 0) [160]; (40, 200) [240]; (40, 240) [280]; (80, 160) [240] & (160, 0) [160] \\ 
  IA & Story & 14 survivors & 5 survivors & (160, 0) [160]; (200, 0) [200]; (280, 0) [280]; (40, 160) [200]; (40, 200) [240] & (160, 0) [160] \\ 
  IL & Shumway & 12 survivors & 1 survivors & (160, 0) [160] & (160, 0) [160] \\ 
  IL & Urbana & 12 survivors & 6 survivors & (240, 0) [240]; (280, 0) [280]; (40, 160) [200]; (40, 200) [240]; (40, 240) [280]; (80, 160) [240] & (40, 160) [200] \\ 
  IN & Loam & 14 survivors & 2 survivors & (40, 240) [280]; (80, 80) [160] & (80, 80) [160] \\ 
  IN & Sand & 14 survivors & 3 survivors & (40, 160) [200]; (40, 240) [280]; (80, 160) [240] & (40, 160) [200] \\ 
  MN & NewRichland & 14 survivors & 2 survivors & (160, 0) [160]; (280, 0) [280] & (160, 0) [160] \\ 
  MN & StCharles & 14 survivors & 3 survivors & (200, 0) [200]; (80, 160) [240]; (80, 80) [160] & (80, 80) [160] \\ 
  MO & Bay & 14 survivors & 1 survivors & (80, 160) [240] & (80, 160) [240] \\ 
  MO & Troth & 14 survivors & 2 survivors & (200, 0) [200]; (280, 0) [280] & (200, 0) [200] \\ 
  ND & Amenia & 15 survivors & 2 survivors & (200, 0) [200]; (80, 160) [240] & (200, 0) [200] \\ 
  ND & Durbin & 15 survivors & 6 survivors & (160, 0) [160]; (200, 0) [200]; (240, 0) [240]; (40, 120) [160]; (40, 80) [120]; (80, 80) [160] & (40, 80) [120] \\ 
  NE & Brandes & 15 survivors & 3 survivors & (40, 120) [160]; (40, 160) [200]; (80, 160) [240] & (40, 120) [160] \\ 
  NE & Kyes & 15 survivors & 3 survivors & (200, 0) [200]; (40, 200) [240]; (40, 240) [280] & (200, 0) [200] \\ 
  WI & Belmont & 15 survivors & 7 survivors & (120, 0) [120]; (160, 0) [160]; (40, 120) [160]; (40, 200) [240]; (40, 40) [80]; (40, 80) [120]; (80, 0) [80] & (40, 40) [80] \\ 
  WI & Steuben & 15 survivors & 2 survivors & (280, 0) [280]; (80, 160) [240] & (80, 160) [240] \\ 
   \hline
\end{tabular}
\endgroup
\end{table}
\end{document}